\RequirePackage[l2tabu, orthodox]{nag}
\documentclass[11 pt]{amsart}
\usepackage{amsmath, amsthm}
\usepackage{amssymb}
\usepackage{graphicx}
\usepackage[utf8]{inputenc}
\usepackage{rotating}
\usepackage{tikz}
\usepackage{rotfloat}
\usepackage{lscape}
\usepackage{setspace}
\usepackage{microtype}
\usepackage{lscape}
\usepackage{setspace}

\newtheorem{theorem}{Theorem}[section]

\theoremstyle{definition}

\theoremstyle{remark}

\numberwithin{equation}{section}

\begin{document}
\bibliographystyle{plain}
\title{Lower bound for the unique games problem}
\author{Rajeev Kohli}
\author{Ramesh Krishnamurti}

\date{July 28, 2015.\\  {\it Address}: R. Kohli, Graduate School of Business, Columbia University, New York, NY (rk35@columbia.edu) and R. Krishnamurti, School of Computing Science, Simon Fraser University, Burnaby, British Columbia, Canada (ramesh@cs.sfu.ca).}

\maketitle
\begin{abstract}
We consider a randomized algorithm for the unique games problem, using independent multinomial probabilities to assign labels to the vertices of a graph. The expected value of the solution obtained by the algorithm is expressed as a function of the probabilities. Finding probabilities that maximize this  expected value is shown to be equivalent to  obtaining an optimal solution to the unique games problem. We attain an upper bound on the optimal solution value by solving a semidefinite programming relaxation of the problem in polynomial time. We use a different but related formulation to show that this upper bound is no greater than $\pi/2$ times the value of the optimal solution to the unique games problem. 

\end{abstract}


\section{Introduction}
Khot's \cite{KH2} unique games conjecture is an important open question in the area of computational complexity. It says that for certain constraint satisfaction problems, called unique games, it is NP-hard to distinguish between instances that are almost satisfiable and almost completely unsatisfiable.  Khot and Vishnoi \cite{KV} discussed how the conjecture has led to connections between computational complexity, algorithms, analysis and geometry. Raghavendra \cite{RAG} showed that, if the conjecture is true, every constraint satisfaction problem has an associated sharp approximation threshold. For background to the problem and the related literature, see Trevisan \cite{TRE}.

Arora et al. \cite{ABS} observed that the unique games conjecture is one of the few open questions that could go either way. It would not be true if a polynomial time procedure obtained a non-trivial lower bound on the optimal solution value for the unique games problem. Polynomial time algorithms obtaining such bounds have been developed for restricted families of the problem, including those whose constraint graphs have expansion properties, are random graphs, or are random geometric graphs  (\cite{AKKSTV}, \cite{BHHS}, \cite{KOL}). For arbitrary problem instances, nontrivial lower bounds can be obtained in subexponential time by using randomized algorithms due to Arora et al.  \cite{ABS} and Boaz et al \cite{BRS}. But the performance of polynomial time algorithms for the general problem deteriorates as the number of  labels increases (e.g., \cite{CMM1}, \cite{CMM2}, \cite{GT}, \cite{KH2}, \cite{TRE1}). 

We describe a polynomial time procedure that does not find a solution but obtains a non-trivial lower bound on the optimal solution value of the unique games problem. The bound is obtained in two steps. The first step develops a continuous formulation of the unique games problem,  then solves its semidefinite programming relaxation in polynomial time to attain an upper bound on the optimal solution value. The second step uses a geometric representation of the continuous problem to show that its optimal solution value is no smaller than $2/\pi$ times the value of this upper bound. The formulation used for the semidefinite programming relaxation generalizes Goemans and Williamson's \cite{GW} representation of the maxcut problem, and is different from previous formulations for the unique games problem (e.g., \cite{KH2}, \cite{CMM1}, \cite{GT}, \cite{GS}, \cite{TRE1}). It is obtained as follows.

Consider a randomized algorithm that assigns labels to vertices using independent multinomial probability distributions. The expected value of its solution is a function of the probabilities with which it assigns labels to vertices. We show that the problem of maximizing this  expected value over the probabilities is equivalent to finding an optimal solution to the unique games problem. We use the probabilistic representation to obtain two different formulations of the problem. Both formulations associate a vector in a unit sphere with a label for a vertex. The first formulation uses the cosines of the angles between vectors, and the second the angles themselves, to  characterize the probabilities associated with the randomized algorithm. An upper bound on the optimal solution value is obtained by solving a semidefinite programming relaxation of the first formulation in polynomial time. The second formulation  is used to show that this upper bound is no greater than $\pi/2$ times the optimal solution value of the unique games problem. 

{\it Organization of the paper}. Section 2 describes the unique games problem, formulates the problem of maximizing the expected value of the randomized algorithm, shows that it is an extension of the unique games problem over a probability space, and discusses its relation with Goemans and Williamson's formulation for the maxcut problem. Section 3 describes a vector representation of the problem and obtains the semidefinite programming relaxation. Section 4 develops the alternative geometric formulation, examines  its relation with the  semidefinite programming relaxation, and obtains a lower bound on the optimal solution value for the unique games problem. 

\section{Unique games problem}
Let $G(V, E)$ denote a graph with $|V|=n$ vertices and $|E|=m$ edges. Each vertex can be assigned one of $k\ge 2$ labels, denoted $r=1, \dots, k$. Each edge $(i, j)\in E$ has weight $w_{ij}>0$. We say that edge $(i, j)\in E$ is matched (equivalently, its vertices are matched) if vertex $i$ is assigned label $r$ and vertex $j$ is assigned label $\sigma_{ij}(r)$, where $r=1, \dots, k$ and $\sigma_{ij}(r)\ne \sigma_{ij}(t)$ if $r\ne t$. Thus, for each edge $(i, j)\in E$, the elements of the ordered vector $\sigma_{ij}=(\sigma_{ij}(1), \dots, \sigma_{ij}(k))$ are the integers  $1, \dots, k$. The $r$th element $\sigma_{ij}(r)$ corresponds to the label for vertex $j$ that matches label $r$ for vertex $i$.  The objective of the unique games problem is to find an assignment of labels to all vertices that maximizes the sum of the weights $w_{ij}$ across matched edges.

\subsection{Formulation}
Consider a randomized algorithm that assigns label $r$ to vertex $i\in V$ with a multinomial probability  $p_{ir}$, where $p_{i1}+\dots + p_{ik}=1$, for all $i\in V$. The probabilities $p_{ir}$ can differ across both the labels and the vertices for a problem instance. We consider the problem of finding the probability values that maximize the expected value of the solution obtained by the randomized algorithm. We show that solving this problem is equivalent to finding an optimal solution for the associated unique games problem. 

The randomized algorithm matches (the vertices $i$ and $j$ of) an edge $(i, j)\in E$ with the probability
$$\rho_{ij}=\sum_{r=1}^k p_{ir}p_{j\sigma_{ij}(r)},\ \text{for\ all}\ (i, j)\in E.$$
Thus, the solution it obtains has the expected value
$$E[z]=\sum_{(i, j)\in E} w_{ij}\rho_{ij}=\sum_{(i, j)\in E} \sum_{r=1}^k w_{ij} p_{ir}p_{j\sigma_{ij}(r)}.$$
We consider the following problem, denoted P, in which the decision variables are the probabilities $p_{ir}$ and the objective function maximizes the value of $E[z]$.

$${\rm (P)}\ \ {\rm Maximize}\ \ E[z]=\sum_{(i, j)\in E} \sum_{r=1}^k w_{ij} p_{ir}p_{j\sigma_{ij}(r)}
\qquad\qquad\qquad\qquad\ \ $$
$${\rm subject\ to}\ \ \sum_{r=1}^k p_{ir}=1,\ \text{for\ all}\ i\in V,\qquad\qquad\qquad\quad\ \ $$
$$\qquad \quad\ \ \  0\le p_{ir}\le 1,\ \text{for\ all}\ r=1, \dots, k,\ i\in V.$$
Since P is a maximization problem, the constraint on the sum of the probabilities can be written as an inequality: $p_{i1}+\dots +p_{ik}\le 1$, for all $i\in V$. We will use this representation for proving Theorem \ref{thm2}. 

Let
$$p_{ir}={1\over 2}(1+y_{ir}),\ {\rm where}\ -1\le y_{ir}\le 1.$$ 
Then $y_{ir}/2$ is the deviation of $p_{ir}$ from  $1/2$, and has a value between $-1/2$ and $1/2$.  (Poljak et al. \cite{PRW} used the same method to convert optimization problems with 0-1 decision variables into those with $\pm 1$ values.) Thus
$$\rho_{ij}={1\over 4}\sum_{r=1}^k (1+y_{ir}+y_{j\sigma_{ij}(r)}+y_{ir}y_{j\sigma_{ij}(r)}),\ {\rm for\ all}\ (i, j)\in E.$$
The constraint  
$$\sum_{r=1}^k p_{ir}=\sum_{r=1}^k {1\over 2}(1+y_{ir})=1$$
becomes
$$\sum_{r=1}^k y_{ir}=2-k,\ \text{for\ all}\ i\in V.$$
Thus, the following problem, denoted P1,  is equivalent to problem P.

$${\rm (P1)}\ \ {\rm Maximize}\ \ E[z]={1\over 4}\sum_{(i, j)\in E} \sum_{r=1}^k w_{ij} (1+y_{ir}+y_{j\sigma_{ij}(r)}+y_{ir}y_{j\sigma_{ij}(r)})$$
$${\rm subject\ to}\ \ \sum_{r=1}^k y_{ir}=2-k,\ \ {\rm for\ all}\ i\in V, \qquad\qquad\qquad\quad\ $$
$$\qquad\quad -1\le y_{ir}\le 1,\ {\rm for\ all}\ r=1, \dots, k,\  i\in V.$$

Let $z$ denote the value of a feasible solution, and $z^*$ the value of the optimal solution, to a unique games problem.

\begin{theorem}\label{thm1}
Problem P1 is a continuous extension of the unique games problem. Its optimal solution (1)  is obtained when $y_{ir}\in \{-1, 1\}$, and (2) has the same value $z^*$ as the optimal solution to the unique games problem.
\end{theorem}

\begin{proof}
To show that the problem of maximizing $E[z]$ is a continuous extension, it is sufficient to observe that (1) it is well-defined for all values of $-1\le y_{ir}\le 1$ (that is, $0\le p_{ir}\le 1$), where $r=1, \dots, k$ and $i\in V$; and (2) any  feasible solution to the unique games problem in which each vertex $i\in V$ is assigned label $r_i$ is also obtained by the randomized algorithm when $y_{ir_i}=1$ (that is, $p_{ir_i}=1$)  and $y_{ir}=-1$ (that is, $p_{ir}=0$)  for each $r\ne r_i$, $r=1, \dots, k$. It follows that 
$\max E[z]\ge z^*$, because the optimal solution to the unique games problem is a feasible solution to the problem of maximizing $E[z]$. On the other hand, $\max E[z]\le z^*$, because $E[z]$ is an expected value computed over the set of feasible solutions to the unique games problem, none of which can exceed the value $z^*$. Thus, $\max E[z]=z^*$.

\end{proof}

Restricting the probabilities $p_{ir}$ to 0-1 values in problem P gives a discrete formulation of  the unique games problem, edge $(i, j)$ being matched with probability $\rho_{ij}\in \{0, 1\}$ for all $(i, j)\in E$. This is equivalent to restricting the $y_{ir}$ variables to $\pm 1$ values in problem P1. We interpret the latter formulation and observe  that Goemans and Williamson's \cite{GW} formulation for the maxcut problem is its special case.

First, note that the constraint $y_{i1}+\dots + y_{ik}=2-k$
in problem P1 implies that
$$\rho_{ij}={1\over 4}\big(4-k+\sum_{r=1}^k y_{ir}y_{j\sigma_{ij}(r)}\big).$$
Since $0\le \rho_{ij}\le 1$, 
$$k-4\le \sum_{r=1}^k y_{ir}y_{j\sigma_{ij}(r)}\le k.$$
Next, suppose $y_{ir}\in \{-1, 1\}$, for all $r=1, \dots, k$ and $i\in V$.  Then the constraint
$y_{i1}+\dots + y_{ik}=2-k$ is  satisfied only if $y_{ir}=1$ for the label that is assigned to vertex $i$, and $y_{ir}=-1$ for the other $k-1$ labels. Suppose that edge $(i, j)$ is matched, that vertex $i$ is assigned label $s$, and that vertex $j$ is assigned the matching label $\sigma_{ij}(s)$. In this case, 
$$y_{is}y_{j\sigma_{ij}(s)}=(1)(1)=1,$$
and 
$$y_{ir}y_{j\sigma_{ij}(r)}=(-1)(-1)=1,\ {\rm for\ all}\ r\ne s.$$ 
Thus,
$$\sum_{r=1}^k y_{ir}  y_{j\sigma_{ij}(r)}=1(1)+(k-1)(1)=k,$$ 
and
$$\rho_{ij}={1\over 4}\big(4-k+\sum_{r=1}^k y_{ir}  y_{j\sigma_{ij}(r)}\big)=1.$$ 
Now suppose that edge $(i, j)$ is not matched, that vertex $i$ is assigned label $s$, and that vertex $j$ is assigned label $\sigma_{ij}(t)$, where $s\ne t$. In this case, 
$$y_{is}y_{j\sigma_{ij}(s)}=(1)(-1)=-1,$$
$$y_{it}y_{j\sigma_{ij}(t)}=(-1)(1)=-1,$$  
and
$$y_{ir}y_{j\sigma_{ij}(r)}=(-1)(-1)=1,\ {\rm for\ all}\ r\ne s, t.$$
Thus,
$$\sum_{r=1}^k y_{ir}  y_{j\sigma_{ij}(r)}=(-1)+(-1)+(k-2)(1)=k-4,$$ 
and
$$\rho_{ij}={1\over 4}\big(4-k+\sum_{r=1}^k y_{ir}  y_{j\sigma_{ij}(r)}\big)=0.$$ 
We conclude that if $y_{ir}\in \{-1, 1\}$, then $\rho_{ij}\in \{0, 1\}$, and $\sum_{i, j)\in E} w_{ij}\rho_{ij}$ is the value of a feasible solution for the unique games problem.

\subsection{Maxcut problem} Consider problem P1 for $k=2$, $y_{j\sigma_{ij}(1)}=2$ and $y_{j\sigma_{ij}(2)}=1$. Then 
$$\rho_{ij}={1\over 4}(4-k+y_{i1}y_{j\sigma_{ij}(1)}+y_{i2}y_{j\sigma_{ij}(2)})$$ 
$$={1\over 4}(2+y_{i1}y_{j2}+y_{i2}y_{j1}).\qquad\quad$$ 
Problem P1 becomes

$${\rm Maximize}\ \ E[z]={1\over 4}\sum_{(i, j)\in E} w_{ij}(2+y_{i1}y_{j2}+y_{i2}y_{j1})$$
$${\rm subject\ to}\ \ y_{i1}+y_{i2}=0,\ \text{for\ all}\ i\in V_1,\qquad\qquad\ $$
$$\qquad\qquad \quad\ \ \  -1\le y_{i1}, y_{i2}\le 1,\ \text{for\ all}\ r=1, 2,\ i\in V.$$

We eliminate the constraint $y_{i1}+y_{i2}=0$ by substituting  $y_{i2}=-y_{i1}$ and $y_{j2}=-y_{j1}$ into the objective function to obtain the following representation:

$${\rm Maximize}\ \ E[z]={1\over 2}\sum_{(i, j)\in E} w_{ij}(1-y_{i1}y_{j1})$$
$${\rm subject\ to}\ \ -1\le y_{i1}\le 1,\ \text{for\ all}\ i\in V.\  $$

This is the formulation described by Goemans and Williamson \cite{GW} for the maxcut problem. The only difference, which is inconsequential after Theorem \ref{thm1}, is that the present formulation maximizes the expected value of a randomized algorithm and allows each $y_{i1}$ variable to obtain any value between $-1$ and $1$.

\section{Vector representation and relaxation}
 Let $S_{(k+1)n}$ denote a unit sphere in $(k+1)n$ dimensions. Let $v_{ir}$ denote a unit vector in $S_{(k+1)n}$,  for each $r=0, \dots, k$ and $i\in V$.  We associate the vector $v_{ir}$  with label $r$ for vertex $i$, for each $r=1, \dots, k$ and $i\in V$. The vectors $v_{i0}$ are not associated with labels, but are used as follows to define the probabilities with which the randomized algorithm assigns  labels to the vertices.
 
Let 
$$y_{ir}=v_{i0}\cdot v_{ir},\ {\rm for\ all}\ r=1, \dots, k,\ i\in V.$$ 
Then the probability that vertex $i$ is assigned label $r$ is given by
$$p_{ir}={1\over 2}(1+v_{i0}\cdot v_{ir}),\ {\rm for\ all}\ r=1, \dots, k,\  i\in V.$$
Equivalently,  
$$v_{i0}\cdot v_{ir}=2p_{ir}-1,\ {\rm for\ all}\ r=1, \dots, k,\ i\in V.$$
Thus, $v_{i0}\cdot v_{ir}=-1$ when $p_{ir}=0$, and $v_{i0}\cdot v_{ir}=1$ when $p_{ir}=1$: the vectors $v_{i0}$ and $v_{ir}$ lie in opposite directions when $p_{ir}=0$, and in the same direction when $p_{ir}=1$. The constraint $p_{i1}+\dots + p_{ik}=1$, which was represented in problem P1 as $y_{i1}+\dots +y_{ik}=2-k$, becomes
$$\sum_{r=1}^k v_{i0}\cdot v_{ir}=2-k,\ \ {\rm for\ all}\ i\in V.$$

Let
$$y_{ir}y_{j\sigma_{ij}(r)}=v_{ir}\cdot v_{j\sigma_{ij}(r)},\ {\rm for\ all}\ r=1, \dots, k,\ i, j\in V.$$ 
Since $y_{ir}y_{j\sigma_{ij}(r)}=(2p_{ir}-1)(2p_{j\sigma_{ij}(r)}-1)$, we have
$$v_{ir}\cdot v_{j\sigma_{ij}(r)}=4\Big(p_{ir}-{1\over 2}\Big)\Big(p_{j\sigma_{ij}(r)}-{1\over 2}\Big).$$
Thus for any label $r=1, \dots, k$:
\begin{itemize}
\item[(1)] $v_{ir}\cdot v_{j\sigma_{ij}(r)}=1$ when (i) $p_{ir}=p_{j\sigma_{ij}(r)}=1$ (that is, $y_{ir}=y_{j\sigma_{ij}(r)}=1$); or (ii) $p_{ir}=p_{j\sigma_{ij}(r)}=0$ (that is, $y_{ir}=y_{j\sigma_{ij}(r)}=-1$). The vectors $v_{ir}$ and $v_{j\sigma_{ij}(r)}$ lie in the same direction when (i) vertices $i$ and $j$ are assigned labels $r$ and $\sigma_{ij}(r)$, respectively, or (ii) when both vertices  are not assigned these labels.
\item[(2)] $v_{ir}\cdot v_{j\sigma_{ij}(r)}=-1$ when $p_{ir}=1$, $p_{j\sigma_{ij}(r)}=0$ (that is, $y_{ir}=1$, $y_{j\sigma_{ij}(r)}=-1$);  or $p_{ir}=0$, $p_{j\sigma_{ij}(r)}=1$ (that is, $y_{ir}=-1$, $y_{j\sigma_{ij}(r)}=1$). The vectors $v_{ir}$ and $v_{j\sigma_{ij}(r)}$ lie in opposite directions  when one, but not both, of the vertices $i$ and $j$ are assigned labels $r$ and $\sigma_{ij}(r)$, respectively. 
\end{itemize}
Since $v_{i0}\cdot v_{i0}=1$, we can express Problem P1 in the  form of problem P2 below.

$${\rm (P2)}\ \ {\rm Maximize}\ \ E[z]={1\over 4}\sum_{(i, j)\in E} \sum_{r=1}^k w_{ij}  (v_{i0}\cdot v_{i0}+v_{i0}\cdot v_{ir}+v_{j0}\cdot v_{j\sigma_{ij}(r)}+v_{ir}\cdot v_{j\sigma_{ij}(r)})$$
$${\rm subject\ to}\ \ \sum_{r=1}^k v_{i0}\cdot v_{ir}\le 2-k,\ \ {\rm for\ all}\ i\in V,\qquad\qquad\qquad\qquad\qquad\qquad\ $$
$$\ \ \ v_{ir}\cdot v_{j\sigma_{ij}(r)}=y_{ir}y_{j\sigma_{ij}(r)},\ {\rm for\ all}\ r=1, \dots, k,\  i\in V,$$
$$\qquad\ \  -1\le y_{ir}\le 1, v_{ir}\in S_{(k+1)n},\ {\rm for\ all}\ r=0, \dots, k,\  i\in V.$$

Observe that we have relaxed the equality constraints on the sum of the probabilities in problem P1 to inequality constraints on the sum of $v_{i0}\cdot v_{ir}$ values in problem P2. Since P2 is a maximization problem, these constraints are tight in the optimal solution. From Theorem \ref{thm1}, the optimal solution to problem P2 is obtained when all the vectors $v_{ir}$ lie in a 1-dimensional space. 
 
Relaxing the constraint $v_{ir}\cdot v_{j\sigma_{ij}(r)}=y_{ir}y_{j\sigma_{ij}(r)}$ in problem P2 gives the following vector program P3, which can be solved in polynomial time. 

$${\rm (P3)}\ \ {\rm Maximize}\ \ z_1={1\over 4}\sum_{(i, j)\in E} \sum_{r=1}^k w_{ij} (v_{i0}\cdot v_{i0}+v_{i0}\cdot v_{ir}+v_{j0}\cdot v_{j\sigma_{ij}(r)}+v_{ir}\cdot v_{j\sigma_{ij}(r)})$$
$${\rm subject\ to}\ \ \sum_{r=1}^k v_{i0}\cdot v_{ir}\le 2-k,\ \ {\rm for\ all}\ i\in V,\qquad\qquad\qquad\quad$$
$$\ \ \ v_{ir}\in S_{(k+1)n},\ {\rm for\ all}\ r=0, \dots, k,\  i\in V.$$
Let $z_1^*$ denote the optimal solution value for problem P3. 

\section{Geometric formulation and a lower bound}
To obtain the desired bound on the optimal solution value of the unique games problem, we obtain another formulation of problem P. This new formulation is closely related to the preceding semidefinite programming relaxation. By construction, the lower bound on the optimal solution value for this new formulation is no smaller than $2/\pi$ times the value of the optimal solution to problem P3. The key difference between the formulation of problem P2 and the following formulation is that while the former represents the $y_{ir}$ variables by the cosines of angles between vectors, the latter represents them by the angles themselves. 

Again, consider a unit sphere $S_{(k+1)n}$, and unit vectors $v_{ir}$ representing labels $r=1, \dots, k$,  for each vertex $i\in V$. As in the preceding formulation, $v_{i0}$ denotes an additional unit vector for each $i\in V$. Let 
$$p_{ir}={1\over 2}(1+y_{ir})=1-{1\over \pi}\arccos(v_{i0}\cdot v_{ir}).$$
Then
$$y_{ir}=2\Big(1-{1\over \pi}\arccos(v_{i0}\cdot v_{ir})\Big)-1.$$
Using the relation $\arcsin (x) + \arccos (x)=\pi/2$ gives
$$y_{ir}={2\over \pi}\arcsin (v_{i0}\cdot v_{ir}),\ {\rm for\ all}\ r=1, \dots, k,\  i\in V.$$

\begin{figure}[!ht]
\caption{Geometric representation of $y_{ir}={2\over \pi}\theta$, where $\theta=\arcsin(v_{i0}\cdot v_{ir})$} 
\centering
\includegraphics[width=0.55\textwidth]{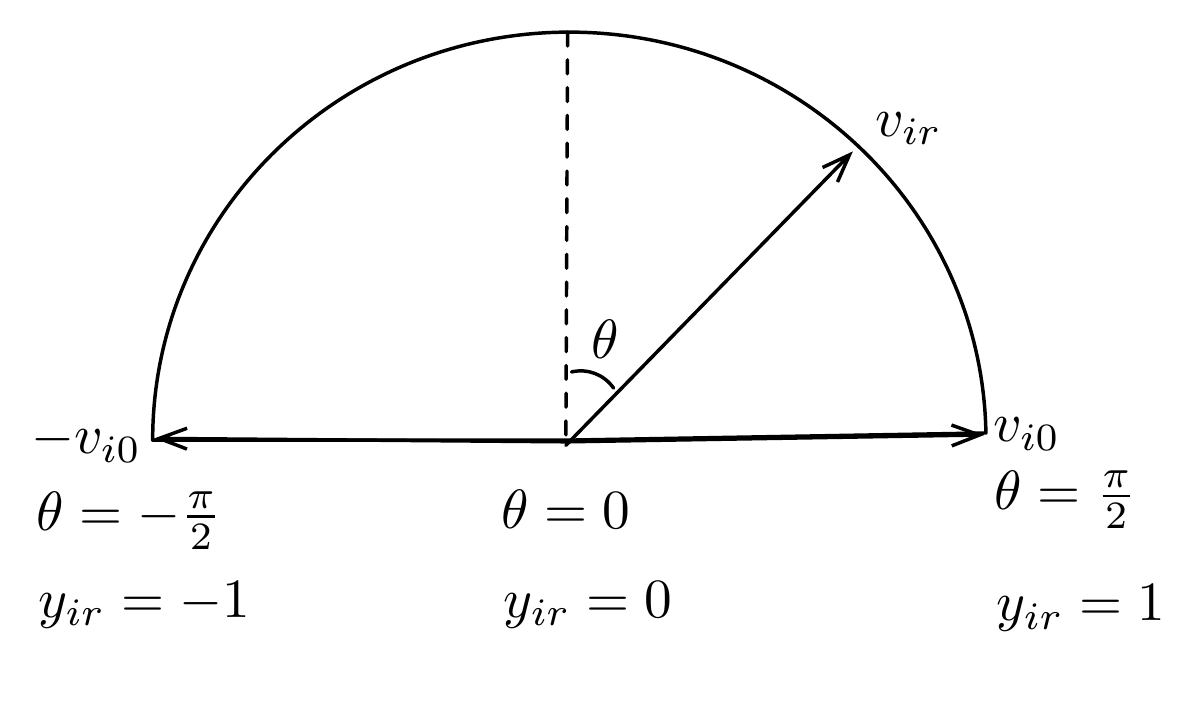}
\end{figure}

Figure 1 shows the relation between $\theta=\arcsin(v_{i0}\cdot v_{ir})$ and $y_{ir}=2p_{ir}-1$. We observe that for any label $r=1, \dots, k$:
\begin{itemize}
\item[(1)] $\theta=\pi/2$ when $y_{ir}=1$: vectors $v_{i0}$ and $v_{ir}$ lie in the same direction when vertex $i$ is assigned label $r$ with probability $p_{ir}=1$.
\item[(2)] $\theta=-\pi/2$ when $y_{ir}=-1$: vectors $v_{i0}$ and $v_{ir}$ lie in opposite directions when vertex $i$ is assigned label $r$ with probability $p_{ir}=0$.
\end{itemize}
The constraint that each vertex is assigned a label with probability one becomes
$$\sum_{r=1}^k {2\over \pi} \arcsin (v_{i0}\cdot v_{ir})=2-k,\ \ {\rm for\ all}\ i\in V.$$

The probability that vertex $i$ is assigned label $r$ but vertex $j$ is not assigned label $\sigma_{ij}(r)$ is given by $p_{ir}(1-p_{j\sigma_{ij}(r)})$. Similarly, the probability that vertex $j$ is assigned label $\sigma_{ij}(r)$ but vertex $i$ is not assigned label $r$ is given by $(1-p_{ir})p_{j\sigma_{ij}(r)}$. Thus, the expression
$$p_{ir, j\sigma_{ij}(r)}=p_{ir}(1-p_{j\sigma_{ij}(r)})+(1-p_{ir})p_{j\sigma_{ij}(r)},\ {\rm for\ all}\ r=1, \dots, k,$$
gives the probability that edge $(i, j)$ is not matched because one, but not both, of vertices $i$ and $j$ are assigned label $r$ and  label $\sigma_{ij}(r)$, respectively. In this case, we say that edge $(i, j)$ is not matched via label $r$. Let $p_{ir, j\sigma_{ij}(r)}$ be proportional to the angle between vectors $v_{ir}$ and $v_{j\sigma_{ij}(r)}$, for each $r=1, \dots, k$:
$$p_{ir, j\sigma_{ij}(r)}=p_{ir}(1-p_{j\sigma_{ij}(r)})+(1-p_{ir})p_{j\sigma_{ij}(r)}={1\over \pi}{\arccos(v_{ir}\cdot v_{j\sigma_{ij}(r)})}.$$ 
Multiplying both sides of the preceding expression by $-2$ gives
$$4p_{ir}p_{j\sigma_{ij}(r)}-2p_{ir}-2p_{j\sigma_{ij}(r)}=-{2\over \pi}\arccos(v_{ir}\cdot v_{j\sigma_{ij}(r)}).$$ 
Adding $1$ to both sides of this expression gives
$$4p_{ir}p_{j\sigma_{ij}(r)}-2p_{ir}-2p_{j\sigma_{ij}(r)}+1=(2p_{ir}-1)(2p_{j\sigma_{ij}(r)}-1)=2\big(1-{1\over \pi}\arccos(v_{ir}\cdot v_{j\sigma_{ij}(r)})\big)-1.$$ 
Since $p_{ir}=(1+y_{ir})/2$, we substitute $2p_{ir}-1=y_{ir}$ to obtain
$$y_{ir}y_{j\sigma_{ij}(r)}=2\big(1-{1\over \pi}\arccos(v_{ir}\cdot v_{j\sigma_{ij}(r)})\big)-1\qquad\qquad\qquad\qquad\qquad\qquad\ $$
$$={2\over \pi}\arcsin (v_{ir}\cdot v_{j\sigma_{ij}(r)}),\ {\rm for\ all}\ r=1, \dots, k,\ i, j\in V.$$

\begin{figure}[!ht]
\caption{Geometric representation of $y_{ir}y_{j\sigma_{ij}(r)}={2\over \pi}\theta$, where $\theta=\arcsin(v_{ir}\cdot v_{j\sigma_{ij}(r)})$} 
\centering
\includegraphics[width=0.67\textwidth]{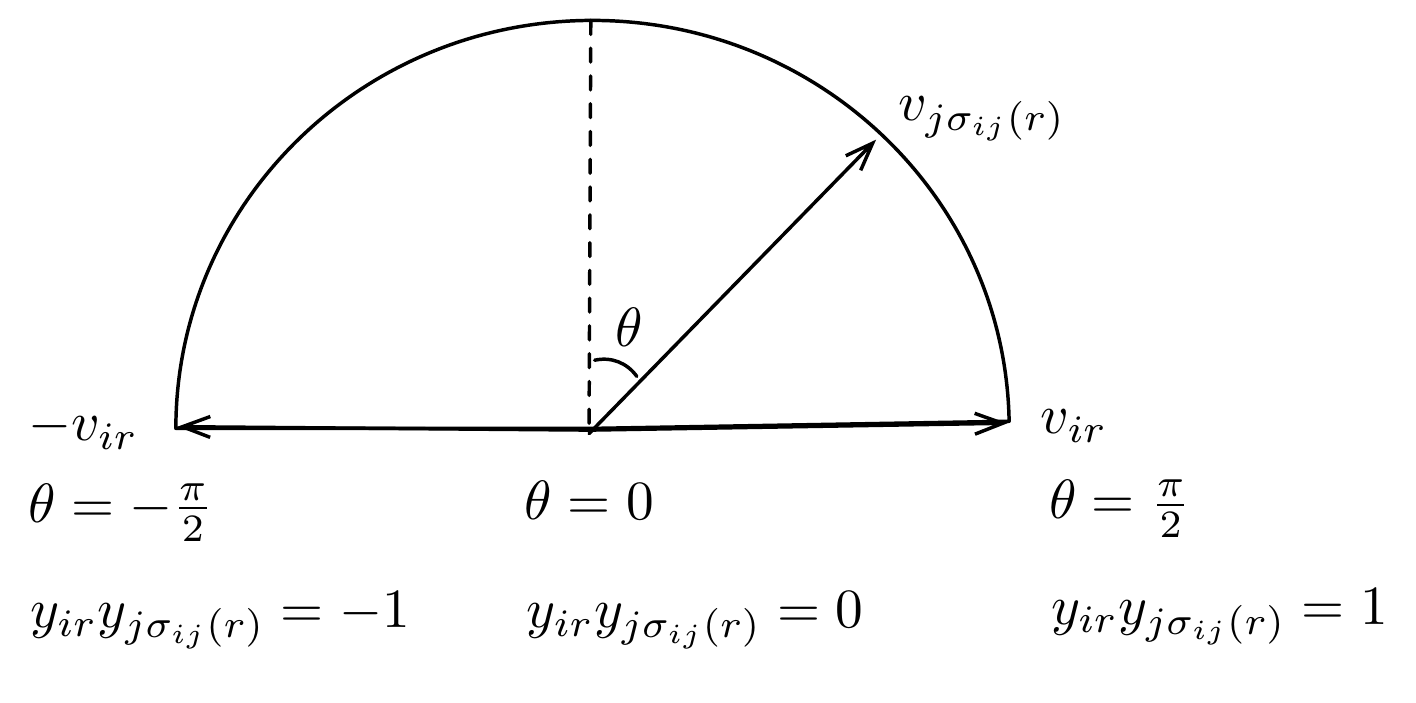}
\end{figure}

Figure 2 shows the relation between $\theta=\arcsin(v_{ir}\cdot v_{j\sigma_{ij}(r)})$ and $y_{ir}y_{j\sigma_{ij}(r)}=(2p_{ir}-1)(2p_{j\sigma_{ij}(r)}-1)$.
We observe that for any label $r=1, \dots, k$:  
\begin{itemize}
\item[(1)] $\theta=\pi/2$ when $y_{ir}=y_{j\sigma_{ij}(r)}=1$ or $y_{ir}=y_{j\sigma_{ij}(r)}=-1$. That is, vectors $v_{ir}$ and $v_{j\sigma_{ij}(r)}$ lie in the same direction when (i) vertices $i$ and $j$ are  matched using label $r$, or (ii) vertex $i$ is not assigned label $r$ and vertex $j$ is not assigned label $\sigma_{ij}(r)$.
\item[(2)] $\theta=-\pi/2$ when $y_{ir}y_{j\sigma_{ij}(r)}=-1$. That is, vectors  $v_{ir}$ and $v_{j\sigma_{j}(r)}$ lie in opposite directions when edge $(i, j)$ is not matched via  label $r$. 
\end{itemize}

Thus, the probability that edge $(i, j)\in E$ is matched has the value
$$\rho_{ij}={1\over 4}\sum_{r=1}^k \Big\{{2\over \pi}\arcsin (v_{i0}\cdot v_{i0})+{2\over \pi}\arcsin (v_{i0}\cdot v_{ir})+{2\over \pi}\arcsin (v_{j0}\cdot v_{j\sigma_{ij}(r)})+{2\over \pi}\arcsin (v_{ir}\cdot {v_{j\sigma_{ij}(r)})}\Big\},$$
where we have substituted 
$$1={2\over \pi}\arcsin (v_{i0}\cdot v_{i0})$$ 
because $v_{i0}\cdot v_{i0}=1$.  
It follows that problem P2 is equivalent to the following problem, denoted P4.

$${\rm (P4)}\ \ {\rm Maximize}\ \ E[z]={2\over \pi}{1\over 4}\sum_{(i, j)\in E} \sum_{r=1}^k w_{ij} \Big\{\arcsin (v_{i0}\cdot v_{i0})+\arcsin (v_{i0}\cdot v_{ir})$$
$$\qquad\qquad\qquad\qquad\qquad\qquad\qquad\quad\qquad\qquad +\arcsin (v_{j0}\cdot v_{j\sigma_{ij}(r)})+\arcsin (v_{ir}\cdot {v_{j\sigma_{ij}(r)})}\Big\}$$
$${\rm subject\ to}\ \ {2\over \pi}\sum_{r=1}^k \arcsin (v_{i0}\cdot v_{ir})\le 2-k,\ \ {\rm for\ all}\ i\in V,\qquad\ \ $$
$$\quad v_{ir}\in S_{(k+1)n},\ {\rm for\ all}\ r=0, \dots, k,\  i\in V.$$

Theorem \ref{thm1} implies that the optimal solution to problem P4 is characterized by the following two conditions:
\begin{itemize}
\item[(1)] $\arcsin (v_{ir}\cdot v_{j\sigma_{ij}(r)})\in \{-\pi/2, \pi/2\}$, which is equivalent to 
$y_{ir}\in \{-1, 1\}$, for all $r=1, \dots, k$ and $i\in V.$
\item[(2)] $\arcsin(v_{i0}\cdot v_{ir})=\pi/2$ and $\arcsin(v_{i0}\cdot v_{it})=-\pi/2$, for all $t\ne r$,  
 $t=1, \dots, k$ and $i\in V$.  This is equivalent to $y_{ir}=1$ for the label assigned to vertex $i\in V$, and $y_{ir}=-1$ for the remaining $k-1$ labels that are not assigned to vertex $i\in V$.
\end{itemize}
Thus, the optimal solution to problem P4 is obtained when the vectors $v_{ir}$ lie in a one-dimensional space, for all $r=0, \dots, k$ and $i\in V$. 

\newpage
\begin{theorem}\label{thm2}
$z^*=\max E[z]\ge {2\over \pi} z_1^*$.
\end{theorem}

\begin{proof}
Consider the following constraint in problem P3:
$$\sum_{r=1}^k v_{i0}\cdot v_{ir}\le 2-k,\ {\rm for\ all}\ i\in V.$$
We can obtain a relaxed version of the constraint by replacing each term on the left hand side by another term that cannot attain a larger value than $v_{i0}\cdot v_{ir}$. We do so below. 

We substitute $x=v_{i0}\cdot v_{ir}$ in the relation $\arcsin (x)/x\le \pi/2$ and rearrange terms to obtain
$${2\over \pi}\arcsin (v_{i0}\cdot v_{ir})\le  v_{i0}\cdot v_{ir}.$$
This gives the following relaxation of the constraint in problem P3:
$${2\over \pi}\sum_{r=1}^k \arcsin (v_{i0}\cdot v_{ir})\le 2-k,\ {\rm for\ all}\ i\in V.$$
Thus  the following problem, denoted P5, is a relaxation of problem P3.

$${\rm (P5)}\ \ {\rm Maximize}\ \ z_2={1\over 4}\sum_{(i, j)\in E} \sum_{r=1}^k w_{ij} (v_{i0}\cdot v_{i0}+v_{i0}\cdot v_{ir}+v_{j0}\cdot v_{j\sigma_{ij}(r)}+v_{ir}\cdot v_{j\sigma_{ij}(r)})$$
$${\rm subject\ to}\ \ {2\over \pi}\sum_{r=1}^k \arcsin (v_{i0}\cdot v_{ir})\le 2-k,\ \ {\rm for\ all}\ i\in V,\qquad\quad$$
$$\quad v_{ir}\in S_{(k+1)n},\ {\rm for\ all}\ r=0, \dots, k,\  i\in V.$$

Let $z_2^*$ denote the optimal solution value for problem P5.  Then $z_2^*\ge z_1^*$, where $z_1^*$ is the optimal solution value for problem P3. 

Consider problem P4. It has the same constraints as problem P5. Since $\arcsin (x)/x \ge 1$, the value of the objective function in problem P4 is no smaller than $(2/\pi) z_2$. Thus, the optimal solution value of problem P4 has the lower bound
$$\max E[z]=z^*\ge {2\over \pi}z_2^*\ge {2\over \pi}z_1^*.$$ 

\end{proof}

Let $z^*=\max E[z]=(1-\epsilon)n$ denote the optimal solution value for a unique games problem. Since $z_1^*\ge z^*$, Theorem \ref{thm2} implies that we can establish the lower bound  $z^*\ge {2\over \pi}(1-\epsilon)n$ by solving problem P2 in polynomial time. Thus, we can  distinguish such a problem from another unique games problem with optimal solution value less  than ${2\over \pi}(1-\epsilon)n$.

\bibliography{ug_july28}

\end{document}